\documentclass[11pt,a4paper]{article}
\usepackage[hidelinks]{hyperref}
\usepackage[left=3cm, right=3cm, top=3cm, bottom=3.5cm]{geometry}
\usepackage[utf8]{inputenc}
\usepackage{amssymb,amsthm,graphics,color,doi}

\let\epsilon\varepsilon

\newcommand{\SH}{\mathcal{S}}
\newcommand{\OP}[1]{\textsc{#1}}

\newcommand{\INSERT}{\OP{insert}}
\newcommand{\MELD}{\OP{meld}}
\newcommand{\FINDMIN}{\OP{find-min}}
\newcommand{\DELETE}{\OP{delete}}
\newcommand{\EXTRACTMIN}{\OP{extract-min}}
\newcommand{\MAKEHEAP}{\OP{make-heap}}
\newcommand{\REDUCE}{\OP{reduce}}

\newcommand{\FIELD}[2]{{#1}(#2)}
\newcommand{\FIELDRM}[2]{\FIELD{\mathrm{#1}}{#2}}
\newcommand{\SUFFIXMIN}[1]{\FIELDRM{suffix\mbox{-}min}{#1}}
\newcommand{\HEAD}[1]{\FIELDRM{head}{#1}}
\newcommand{\RANK}[1]{\FIELDRM{rank}{#1}}
\newcommand{\KEY}[1]{\FIELDRM{key}{#1}}
\newcommand{\VALUE}[1]{\FIELDRM{value}{#1}}

\newcommand{\C}[1]{\FIELD{C}{#1}}
\newcommand{\W}[1]{\FIELD{W}{#1}}

\newcommand{\floor}[1]{\left\lfloor #1 \right\rfloor}
\newcommand{\ceil}[1]{\left\lceil #1 \right\rceil}
\newcommand{\latin}[1]{\textit{#1}}

\newtheorem{theorem}{Theorem}
\newtheorem{lemma}{Lemma}

\begin{document}

\title{Soft Sequence Heaps}

\author{Gerth Stølting Brodal \\
       Department of Computer Science, Aarhus University, Denmark \\
       \texttt{gerth@cs.au.dk}}

\maketitle

\begin{abstract}
    Chazelle~\cite{Chazelle00} introduced the \emph{soft heap} as a building block for efficient minimum spanning tree algorithms, and recently Kaplan \latin{et al.}~\cite{KaplanKZZ19} showed how soft heaps can be applied to achieve simpler algorithms for various selection problems. A soft heap trades-off accuracy for efficiency, by allowing $\epsilon N$ of the items in a heap to be \emph{corrupted} after a total of $N$~insertions, where a corrupted item is an item with artificially increased key and $0 < \epsilon \leq \frac{1}{2}$ is a fixed error parameter. Chazelle's soft heaps are based on binomial trees and support insertions  in amortized $O(\lg\frac{1}{\epsilon})$ time and extract-min operations in amortized $O(1)$ time. 
    
    In this paper we explore the design space of soft heaps. The main contribution of this paper is an alternative soft heap implementation based on merging sorted sequences, with time bounds matching those of Chazelle's soft heaps. We also discuss a variation of the soft heap by Kaplan \latin{et al.}~\cite{KaplanTZ13}, where we avoid performing insertions lazily. It is based on ternary trees instead of binary trees and matches the time bounds of Kaplan \latin{et al.}, i.e.\ amortized $O(1)$ insertions and amortized $O(\lg\frac{1}{\epsilon})$ extract-min. Both our data structures only introduce corruptions after extract-min operations which return the set of items corrupted by the operation.
\end{abstract}

%%%%%%%%%%%%%%%%%%%%%%%%%%%%%%%%%%%%%%%%%%%%%%%%%%%%%%%%%%%%%%%%%%%%%%%%%%%%    
\section{Introduction}
%%%%%%%%%%%%%%%%%%%%%%%%%%%%%%%%%%%%%%%%%%%%%%%%%%%%%%%%%%%%%%%%%%%%%%%%%%%%

Chazelle in 1998~\cite{Chazelle98} introduced the \emph{soft heap} as a heap data structure surpassing the comparison lower bounds of heaps by allowing the controlled \emph{corruptions} of keys, i.e.\ artificially increasing the values of the keys of a fraction of the inserted items.
The power of soft heaps was demonstrated by Chazelle in~\cite{Chazelle00a}, who showed how soft heaps could be the key ingredient to compute a minimum spanning tree in time $O(m\cdot\alpha(m,n))$, where $\alpha$ is the inverse of Ackermann's function, and $n$ and $m$ are the number of vertices and edges in the graph, respectively.
Pettie and Ramachandran~\cite{PettieR00,PettieR02} subsequently achieved an optimal comparison based minimum spanning tree algorithm, also using soft heaps, with running time matching the (still unknown) decision-tree complexity of the problem.
20~years later the soft heap paper by Chazelle~\cite{Chazelle98} was awarded the ESA Test-of-Time Award 2018 for its significance on the development of algorithms for the fundamental minimum spanning tree problem.
 
Even though soft heaps were crucial for advancing the knowledge on the minimum spanning tree problem, their applications have remained surprisingly sparse in the literature otherwise. One could speculate this is due to their unconventional interface. Kaplan~\latin{et~al}.~\cite{KaplanKZZ19} recently presented new applications of soft heaps, and in particular strengthened the requirements for the interface to the soft heap operations, to 1) report when an item is considered corrupted internal to a soft heap, 2) to tag returned items if they are corrupted, and 3) restrict corruptions to only be allowed \emph{after} the extraction of the current minimum from a soft heap. This modified interface allowed in particular a very simple and elegant solution to the binary heap selection problem, a significant simplification compared to the previous significantly more complex solution by Frederickson~\cite{Frederickson93}.

In this paper we explore the design-space of soft heaps. The goal of this paper is to present an alternative and simple implementation of soft heaps supporting the interface of Kaplan \latin{et al}.~\cite{KaplanKZZ19}. In~\cite{KaplanKZZ19} it was described how the soft heap in~\cite{KaplanTZ13} could support this interface with minor changes. Our solution is based on merging sorted sequences as opposed to all previous solutions which are all based on heap ordered trees. Similar to all previous solutions, our solution also makes essential use of Chazelle's car-pooling idea.

%%%%%%%%%%%%%%%%%%%%%%%%%%%%%%%%%%%%%%%%
\subsection{Soft heaps}
\label{sec:soft-heap-interface}
%%%%%%%%%%%%%%%%%%%%%%%%%%%%%%%%%%%%%%%%

Like a normal priority queue, soft heaps store a set of (key, value) pairs called \emph{items}, where the keys are from an ordered universe.
As opposed to a normal priority, soft heaps are allowed to \emph{corrupt} the keys of the items by artificially increasing the keys. 
A soft heap trades-off accuracy for efficiency, by allowing up to $\epsilon N$ of the items in a heap to be corrupted after a total of $N$~insertions, where $0 < \epsilon \leq \frac{1}{2}$ is a fixed error parameter. Note that the number of allowed corruptions in the heap is $\epsilon N$, which can be larger than the current number of items~$n$ in the soft heap. In particular it is possible that all keys in a soft heap are corrupted when $\epsilon N \geq n$.

We call the original key of an item the \emph{real} key and the increased key the \emph{current} key.
A corrupted key can be increased multiple times by the soft heap, but never lowered from its current key.
When performing a sequence of insertions and extract-min operations on a soft heap, the soft heap always returns items correctly with respect to their current keys. 
The effect of corruptions on the extracted sequence is that an item that gets corrupted internally in the soft heap raises a flag that the extraction of the item may appear later in the sequence of extractions because of the artificially high key, i.e.\ the user might miss out on getting this item extracted in the correct order. 
When a corrupted item eventually is extracted from a soft heap its current key is an upper bound on its real key, and the extraction allows space for another item to get corrupted in the soft heap.

In this paper we adopt the soft heap interface described Kaplan \latin{et al}.~\cite{KaplanKZZ19}, that explicitly notifies the user about the corruptions introduced. For an application of this interface we refer the reader to the elegant heap selection algorithm in~\cite[Section 3]{KaplanKZZ19}.

\begin{itemize}
\item \MAKEHEAP$()$ creates an empty soft heap~$\SH$ and returns a reference to~$\SH$. 
\item \INSERT$(\SH, e)$ inserts item~$e=(k,v)$ with real key~$k$ and value~$v$ into soft heap~$\SH$.
\item \MELD$(\SH_1, \SH_2)$ melds the soft heaps~$\SH_1$ and $\SH_2$, and returns a reference to the resulting soft heap.
\item \FINDMIN$(\SH)$ returns a pair $(e, k)$, where $e$ is an item with minimum current key~$k$ in the soft heap $\SH$.
\item \EXTRACTMIN$(\SH)$ removes an item~$e$ from the soft heap~$\SH$, where $e$ has minimum current key $k$ before the operation, and returns the triple $(e, k, C)$, where $C$ is the list of items in the heap that were not corrupted before $e$ was removed from~$\SH$, but became corrupted as a result of removing $e$ from~$\SH$.
\item \DELETE$(\SH, e)$ removes item~$e$ from the soft heap~$\SH$. Returns a list $C$ of the items where the key became corrupted by removing~$e$. Requires a reference is given to the location of $e$ in the soft-heap.
\end{itemize}

Chazelle~\cite{Chazelle98,Chazelle00} presented the first implementation of soft heaps, by adopting the idea of \emph{car-pooling} to binomial trees, achieving {\INSERT} in amortized time $O(\lg\frac{1}{\epsilon})$ and all other operations in amortized constant time.%
\footnote{$\lg n$ denotes the binary logarithm of $n$}
Kaplan and Zwick in \cite{KaplanZ09} gave a simplified construction based on binary trees with matching amortized performance. 
Kaplan \latin{et al.}~\cite{KaplanTZ13} presented a solution where all operations are amortized constant time except for {\EXTRACTMIN} and {\DELETE} which take amortized $O(\lg\frac{1}{\epsilon})$ time, i.e.\ postponing the dependence on $\epsilon$ to deletions.
All these solutions, like ours, use car-pooling to achieve their efficiency. Essentially car-pooling treats a pool of items as a single item, and assigns all the items in the pool current key equal to the maximum real key in the pool. By appropriately maintaining a collection of pools the bound on the total number of corruptions can be guaranteed within the stated time bounds. 

For other models of computation, Thorup \latin{et al.}~\cite{ThorupZZ19} presented non-comparison based soft heaps for the RAM model achieving amortized $O(\lg\lg\frac{1}{\epsilon})$ time per operation, or amortized expected $O(\sqrt{\lg\lg \rule{0ex}{2.1ex}\smash{\frac{1}{\epsilon}}})$ using randomization. 
Bhushan and Gopalan~\cite{BhushanG12} considered soft heaps in external memory,
achieving amortized $O(\frac{1}{B}\log_{M/B} \frac{1}{\epsilon})$ I/Os per insertion, and other operations in non-posititive amortized I/Os, where $M$ is the main memory size and $B$ the disk block size, provided $N=O(Bm^{M/2(B+\sqrt{m})})$ where $m=M/B$.

%%%%%%%%%%%%%%%%%%%%%%%%%%%%%%%%%%%%%%%%
\subsection{Applications of soft heaps}
%%%%%%%%%%%%%%%%%%%%%%%%%%%%%%%%%%%%%%%%

The groundbreaking applications of soft heaps are in the mentioned minimum spanning tree algorithms by Chazelle~\cite{Chazelle00a} and Pettie and Ramachandran~\cite{PettieR02}. Further applications were given by Chazelle~\cite{Chazelle00} who showed how soft heaps can lead to alternative solutions for computing exact and approximate medians in linear time, yielding an alternative solution to the classical selection algorithm by Blum \latin{et al.}~\cite{BlumFPRT73}, 
algorithms for finding dynamic  percentiles,
and approximate sorting algorithms with running time $O(n\lg\frac{1}{\epsilon})$ generating sequences with at most $\epsilon n^2$ inversions or where each element is assigned a rank within $\epsilon n$ of its true rank. 

Kaplan \latin{et al.}~\cite{KaplanKZZ19} give further applications of soft heaps. Their main contribution is a very simple algorithm to select the $k$-th smallest item in a binary heap in time~$O(k)$, significantly simplifying the previous approach by Frederickson~\cite{Frederickson93} that was achieved over a sequence of improvements starting with running time $O(k\lg k)$, and then adding ideas to first improve this to $O(k\lg\lg k)$, then to $O(k3^{\lg^* k})$, $O(k2^{\lg^* k})$, and finally $O(k)$. Kaplan~\latin{et~al.}\ then apply the heap selection algorithm to develop various new selection algorithms:
an algorithm for selecting the $k$-th smallest item from a row-sorted matrix with $m$ rows in time $O(m\lg\frac{k}{m})$,
matching a previous bound by Frederickson and Johnson~\cite{FredericksonJ82}, 
and a new algorithm with output sensitive running time of $O(m+\sum_{i=1}^m\lg(k_i+1))$, where $k_i$
is the number of items in the $i$-th row smaller than the $k$-th smallest item, and finally an algorithm to find the $k$-th smallest element from $X+Y$, where $X$ and $Y$ are two unordered sets of $m$ and $n$ items respectively, where $m\geq n$, with running time $O(m\lg\frac{k}{m})$ matching a previous bound of Frederickson and Johnson~\cite{FredericksonJ82}. 
Chakrabarti \latin{et al.}~\cite{ChakrabartiPG11} used soft heaps in an experimental study on graph conductance search.

%%%%%%%%%%%%%%%%%%%%%%%%%%%%%%%%%%%%%%%%
\subsection{Results}
%%%%%%%%%%%%%%%%%%%%%%%%%%%%%%%%%%%%%%%%

The main contribution of this paper is a new implementation of soft heaps, \emph{soft sequence heaps}, designed to satisfy the interface of Kaplan \latin{et al.}~\cite{KaplanKZZ19}.

\begin{theorem}
\label{thm:soft-sequence-heap}
    A soft sequence heap supports {\INSERT} in amortized $O(\lg\frac{1}{\epsilon})$ time and all other soft heap operations in amortized constant time, for a fixed error parameter $0<\epsilon<1$.
    After a total of $N$ insertions the soft heap contains at most $\epsilon N$ items with corrupted keys.
\end{theorem}

A (non-soft) sequence heap is a simple priority queue storing its items in a logarithmic number of sorted sequences (see Section~\ref{sec:sequence-heaps}).
In the literature several priority queues exist based on this idea. 
Examples are
external memory priority queues~\cite{BrengelCFM00,BrodalK98}, 
cache efficient priority queues~\cite{Sanders00}, 
and efficient RAM priority queues~\cite{Thorup00}.
Sanders~\cite{Sanders00} coined such an approach a \emph{sequence heap}.
Earlier, Fischer and Paterson~\cite{FischerP94} developed a priority queue aimed at sequential storage also consisting of a sequence of sorted lists.
Our contribution is to adapt the car-pooling idea of Chazelle to sequence heaps.

Table~\ref{tab:previous_work} summarizes our contributions and contains a comparison of the essential properties of our contributions to previous work.

\begin{table}[t]
    \newcommand{\RAISE}[1]{\raisebox{-1.5ex}[0ex][0ex]{#1}}
    \centering
    \caption{Previous and new results for Soft Heaps}
    \begin{tabular}{lcccl}
        && Insert & ExtractMin & \\
        \hline
        Chazelle 2000 \cite{Chazelle00} & & & & Binomial trees\\
        Kaplan, Zwick 2009 \cite{KaplanZ09} & \rlap{\smash{$\left.\rule{0em}{5ex}\right\}$}}
        & $O(\lg\frac{1}{\epsilon})$ & $O(1)$ & Binary trees \\
        \textbf{New} (Section~\ref{sec:soft-sequence-heaps}) & & & & Sorted sequences \\
        \hline
        Kaplan, Tarjan, Zwick 2013 \cite{KaplanTZ13} &
        \rlap{\RAISE{$\left.\rule{0em}{2.5ex}\right\}$}}
        & \RAISE{$O(1)$} & \RAISE{$O(\lg\frac{1}{\epsilon})$} & Binary trees \\ 
        \textbf{New} (Section~\ref{sec:ternary}) && & & Ternary trees \\
        \hline
    \end{tabular}
    \label{tab:previous_work}
\end{table}

%%%%%%%%%%%%%%%%%%%%%%%%%%%%%%%%%%%%%%%%
\subsection{Structure of paper}
%%%%%%%%%%%%%%%%%%%%%%%%%%%%%%%%%%%%%%%%

In Section~\ref{sec:sequence-heaps} we recall the basic idea of (non-soft) sequence heaps. 
In Section~\ref{sec:soft-sequence-heaps} we show how to convert sequence heaps into soft-heaps using car-pooling. 
In Section~\ref{sec:ternary} we discuss a variation of the soft-heap presented by Kaplan \latin{et al.}~\cite{KaplanTZ13} and show that we can satisfy the interface of Kaplan \latin{et al.}~\cite{KaplanKZZ19} \emph{without} buffering insertions.

%%%%%%%%%%%%%%%%%%%%%%%%%%%%%%%%%%%%%%%%
\section{Sequence heaps}
\label{sec:sequence-heaps}
%%%%%%%%%%%%%%%%%%%%%%%%%%%%%%%%%%%%%%%%

A (non-soft) sequence heap stores items in a logarithmic number of sorted sequences $L_1,L_2,\ldots,L_{\ell}$, where each sequence $L_i$ is assigned a non-negative integer \emph{rank} $\RANK{L_i}$. 
The sequences are maintained in a list $\mathcal{L}$ in increasing rank order.
\INSERT$(e)$ creates at the front of $\mathcal{L}$ a new rank zero sequence containing~$e$, 
and repeatedly merges the first two sequences of $\mathcal{L}$ if they have equal rank~$r$ to a new sequence of rank~$r+1$ until all sequences have distinct ranks. 
{\EXTRACTMIN} finds the sequence where the first item has minimum key, and removes and returns this item.
Figure~\ref{fig:sequence-heap} shows the result of applying {\INSERT} and {\EXTRACTMIN} to a sequence heap.

That {\INSERT} and {\EXTRACTMIN} take amortized $O(\lg N)$ time follows from some simple observations: 
A sequence of rank $r$ contains $2^r$ items (if also counting deleted items),
i.e.\ the maximum rank of a sequence after $N$~insertions is at most $\floor{\lg N}$; 
an inserted item can at most participate in a number of merges bounded by the maximum rank;
and since insertions ensure that the sequences have distinct rank the time for extracting the minimum is also bounded by the maximum rank.

\begin{figure}[t]
    \centering
    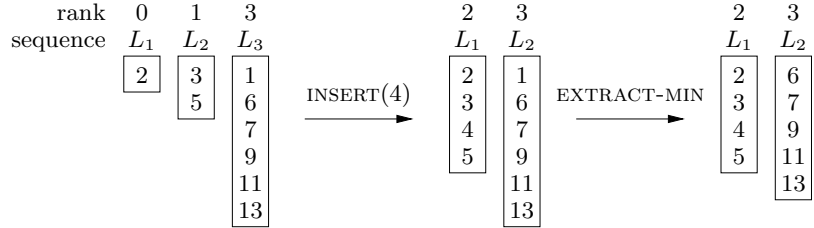
    \caption{A (non-soft) sequence heap. \INSERT$(4)$ first creates a new rank zero sequence~$(4)$, that will we merged with the rank zero sequence~$(2)$, creating the rank one sequence~$(2,4)$, that will be merged with the rank one sequence~$(3,5)$, finally creating the rank two sequence~$(2,3,4,5)$. {\EXTRACTMIN} removes the smallest item, here~1, from the head of its sequence.}
    \label{fig:sequence-heap}
\end{figure}

\begin{figure}[t]
    \centering
    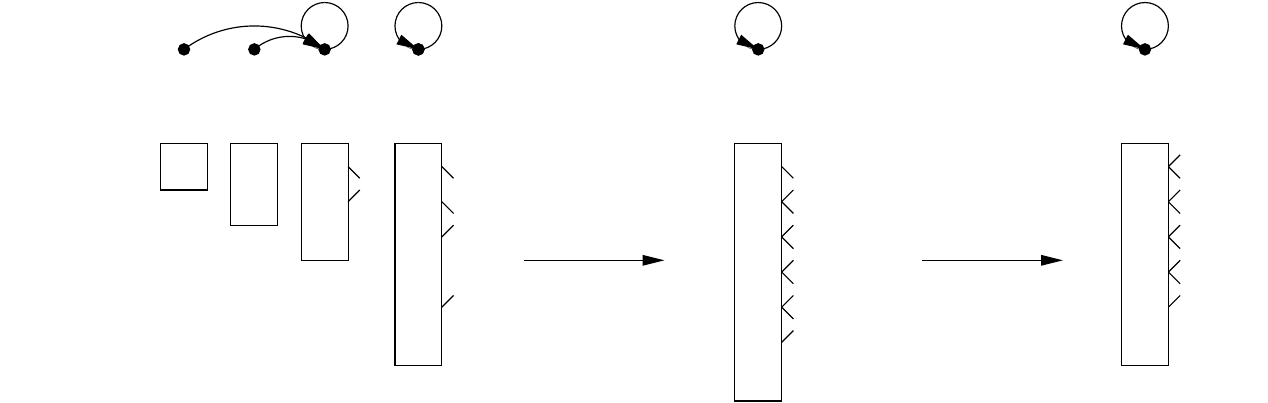
    \caption{A soft sequence heap with $r_0=0$. To the right of item~$e$, $\C{e}$ is shown top-right and $\W{e}$ bottom-right (if non-empty). To perform
    \INSERT$(10)$, a new sequence $(10)$ of rank zero is first merged with rank zero sequence~$(15)$, and then with rank one sequence $(12,14)$ to produce rank two sequence~$(10,12,14,15)$. The pruning of this sequence moves 12 to $\W{10}$ and $\C{14}$. Rank two sequences $(10,14,15)$ and $(3,20,24)$ are then merged to yield a rank three sequence $(3,10,14,15,20,24)$, that finally is merged with $(4,7,18,19,21,23)$, where items $4,10,15,19,21$ are pruned from the resulting rank four sequence. {\EXTRACTMIN} returns the minimum item~3 in the single sequence (since $\C{3}=\emptyset$), and reports $6,16,4$ as corrupted (their current keys are 20, 23, and~7, respectively).}
    \label{fig:soft-sequence-heap}
    \bigskip
    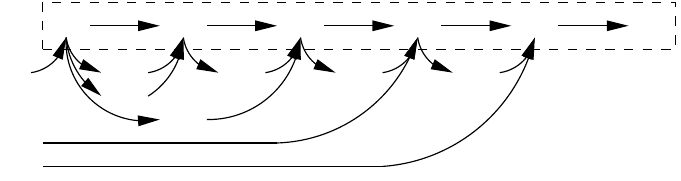
    \caption{The partial order represented by the rightmost sequence in Figure~\ref{fig:soft-sequence-heap}. The items with corrupted keys are items 4, 6 and 16.}
    \label{fig:partial-order}
\end{figure}

%%%%%%%%%%%%%%%%%%%%%%%%%%%%%%%%%%%%%%%%%%%%%%%%%%%%%%%%%%%%%%%%%%%%%%%%%%%%
\section{Soft sequence heaps}
\label{sec:soft-sequence-heaps}
%%%%%%%%%%%%%%%%%%%%%%%%%%%%%%%%%%%%%%%%%%%%%%%%%%%%%%%%%%%%%%%%%%%%%%%%%%%%

In this section we describe  \emph{soft sequence heaps} 
derived by adapting Chazelle's car-pooling idea to \emph{sequence heaps}. 
Below, we first describe the basic ideas used to convert sequence heaps into soft-heaps, 
next we give the details of the representation and the implementation of the operations, 
and finally we analyse our construction.

%%%%%%%%%%%%%%%%%%%%%%%%%%%%%%%%%%%%%%%%
\subsection{Corruption-sets and witness-sets} 
%%%%%%%%%%%%%%%%%%%%%%%%%%%%%%%%%%%%%%%%

To make sequence heaps achieve the performance of soft heaps we essentially adopt two ideas. In the following $0 < \epsilon < 1$ and $r_0 = \ceil{\lg \frac{1}{\epsilon}}$ is a rank threshold.

\emph{Corruption-sets}: With each item~$e$ in a sorted sequence we store a \emph{corruption-set} $\C{e}$ containing items where the key eventually should be raised to $\KEY{e}$.
Whenever the merging of two sequences of equal rank~$r-1$ results in a new sequence of rank~$r > r_0$, where $r-r_0$ is \emph{even},
we \emph{prune every second item from the sequence} (the first and last items in a sequence are not be pruned). 
For an item~$e$ to be pruned and with successor~$e'$ in the sequence, 
we add $e$ to~$\C{e'}$ together with all items from~$\C{e}$. 
This implements the car-pooling idea of Chazelle.

\emph{Witnesses}:
The above unfortunately only leaves $o(N)$ items not pruned from the sequences (see Lemma~\ref{lem:total-size-L} below).
To avoid reporting too many corruptions we postpone reporting items as corrupted until they can influence the order of the items returned by {\EXTRACTMIN}. 
For this purpose we assign a \emph{witness} to each item when it is initially added to a corruption-set.
An item with a witness is not considered corrupted.
When we prune an item~$e$ from a sequence its predecessor $e''$ in the sequence becomes the witness for~$e$.
A corrupted item $e$ in a soft sequence heap is an item $e$ in a corrupted set $\C{e'}$ without a witness. The current key of $e$ is then the real key of $e'$.
For an item $e''$ we let the \emph{witness-set} $\W{e''}$ be all the items $e''$ is a witness for.
When we prune $e$, we add $e$ to the corruption-set $\C{e}$ of its successor $e'$ and to the witness-set $\W{e''}$ of its predecessor.
To $\W{e''}$ we also add all items from $\W{e}$, i.e.\ these items get their witness~$e$ replaced by~$e''$, where $\KEY{e''} \leq \KEY{e}$.
In general, the witness of an item~$e$ is an item $e''$ still in the sequence with $\KEY{e''} \leq \KEY{e}$. 
When an item~$e''$ is deleted from the soft sequence heap we report all items in $\W{e''}$ as corrupted.

%%%%%%%%%%%%%%%%%%%%%%%%%%%%%%%%%%%%%%%%
\subsection{The representation details}
%%%%%%%%%%%%%%%%%%%%%%%%%%%%%%%%%%%%%%%%

We let $e$ denote an item in the heap, 
$\KEY{e}$ the real key of~$e$, 
and $\VALUE{e}$ the value of~$e$. 
A soft sequence heap $\SH$ is represented by 
a list $\mathcal{L}$ of non-empty sequences $L_1, L_2, \ldots, L_\ell$ of items.
Each sequence~$L_i$ has a rank, $\RANK{L_i}$,
the sequences appear in strictly increasing rank order, i.e.\ $\RANK{L_i} < \RANK{L_{i+1}}$ for $1\leq i <\ell$.
The items in $L_i$ are sorted in increasing order by key.
With each item~$e$ in $L_i$ we store a corruption-set $\C{e}$ and witness-set $\W{e}$, possibly empty, 
of items pruned from the sorted sequences but still in the heap.
Both sets are represented by cyclic linked lists, with entry points to the last items in the lists.

Each item~$e$ is stored in exactly one $L_i$ sequence or one corruption-set $\C{e'}$. If $e\in \C{e'}$ then $\KEY{e} \leq \KEY{e'}$, and $e$ is possibly also stored in one witness-set $\W{e''}$, where $\KEY{e''} \leq \KEY{e}$. If $e \in \W{e''}$ and $e \in \C{e'}$, then $e'$ and $e''$ are in the same $L_i$ sequence and $e''$ occurs before $e'$ in the sequence.
The corrupted items in a sequence are precisely the items contained in a corruption-set but not in a witness-set, and if a corrupted item $e\in\C{e'}$ then the current key of $e$ is $\KEY{e'}$.
A sequence together with its corruption-sets and witness-sets can be viewed as maintaining a partial order, where $e'' \leq e$ if $e\in\W{e''}$, and $e\leq e'$ if $e\in\C{e'}$.
See Figure~\ref{fig:partial-order} for an example.

To efficiently maintain a reference to the current minimum item, for each sequence~$L_i$, we maintain a \emph{suffix-min} reference. The same idea was used by Chazelle~\cite{Chazelle00} and implicitly by Kaplan \latin{et al.}~\cite{KaplanTZ13}.
For sequence~$L_i$, $\SUFFIXMIN{L_i}$ is a reference to sequence~$L_j$, where $i\leq j\leq \ell$, such that the first item in $L_j$ has smallest key among the items in $L_i \cup L_{i+1} \cup \cdots \cup L_\ell$.
By definition $\SUFFIXMIN{L_\ell}=L_\ell$ and the first item in the sequence $\SUFFIXMIN{L_1}$ has smallest key among all items in all sequences.
The reference $\SUFFIXMIN{L_i}$ can be updated as follows, assuming $\SUFFIXMIN{L_{i+1}}$ is known:
If $i=\ell$ or $\KEY{\HEAD{L_i}} \leq \KEY{\HEAD{\SUFFIXMIN{L_{i+1}}}}$
then $\SUFFIXMIN{L_i}=L_i$,
otherwise $\SUFFIXMIN{L_i}=\SUFFIXMIN{L_{i+1}}$.
Here $\HEAD{L_i}$ refers to the first item in the sequence.

%%%%%%%%%%%%%%%%%%%%%%%%%%%%%%%%%%%%%%%%
\subsection{Soft heap operations}
%%%%%%%%%%%%%%%%%%%%%%%%%%%%%%%%%%%%%%%%

We now describe how to implement the operations on a soft sequence heap. Given an error parameter $\epsilon$, the rank threshold $r_0=\ceil{\lg\frac{1}{\epsilon}}$ allows us to trade accuracy for improved running time. 
Sequences with rank $\leq r_0$ behave as in a (non-soft) sequence heap and all corruption-sets and witness-sets are empty in these sequences.

The helper method $\REDUCE(L)$ takes a sorted sequence of items $L=e_1,e_2,\ldots,e_m$, 
and prunes~$e_{2i}$ from $L$, for all $1\leq i< m/2$, 
i.e.\ the first and last items are not pruned and
the reduced sequence has length $\ceil{\frac{m + 1}{2}}$. 
Before pruning an item $e_{2i}$ from $L$,
the items $\{ e_{2i} \} \cup \C{e_{2i}}$ are appended to $\C{e_{2i+1}}$, 
and $\{ e_{2i} \} \cup \W{e_{2i}}$ are appended to $\W{e_{2i-1}}$, and $\C{e_{2i}}$ and $\W{e_{2i}}$ cease to exist.
Since $\C{e_{2i}}$ and $\W{e_{2i}}$ sets are cyclic linked lists, this can be done in constant time for each $e_{2i}$.

To support {\DELETE} operations we apply \emph{lazy} deletions, where an item is only marked as being deleted, and remains in the soft sequence heap until it becomes the minimum of the soft sequence heap, where it can be deleted using {\EXTRACTMIN}.
We maintain the invariant that the current minimum item of the soft sequence heap is never an item that has been lazily deleted.

Figure~\ref{fig:soft-sequence-heap} illustrates {\INSERT} and {\EXTRACTMIN} on a soft sequence heap.

\begin{itemize}
    \item \MAKEHEAP$()$ Creates an empty soft heap $\SH$ with $\mathcal{L}=\emptyset$.

    \item \FINDMIN$(\SH)$ Let $L_i=\SUFFIXMIN{L_1}$, $e=\HEAD{L_i}$ and $k=\KEY{e}$.
        Return $(e, k)$ if $\C{e}=\emptyset$, otherwise return $(e', k)$, where $e'=\HEAD{\C{e}}$.

    \item \INSERT$(\SH, e)$ First we create a new sequence of rank zero only containing $e$. This sequence is added to the front of $\mathcal{L}$. While the two first sequences of $\mathcal{L}$ have equal rank~$r$, we merge the two sequences into a sequence of rank $r+1$, that replaces the two first sequences in~$\mathcal{L}$. 
    Whenever creating a sequence $L$ of rank $r>r_0$ where $r-r_0$ is even, we apply $\REDUCE(L)$. Finally we update $\SUFFIXMIN{L_1}$ for the new first sequence $L_1$.

    \item \MELD$(\SH_1, \SH_2)$ Let $\mathcal{L}_1$ and $\mathcal{L}_2$ be the two lists of sequences respectively, and 
    let $r_1$ and  $r_2$ be the maximal ranks of a sequence in $\mathcal{L}_1$  and $\mathcal{L}_2$, respectively. 
    Merge $\mathcal{L}_1$ and $\mathcal{L}_2$ by non-decreasing rank until one of the lists is empty in time $O(\min(r_1, r_2))$. 
    Let  $\mathcal{L}$ be the resulting list. 
    While two sequences in $\mathcal{L}$ have equal rank, merge the two last sequences of equal rank~$r$ and apply $\REDUCE$ if the resulting sequence has rank $r + 1 > r_0$ and $(r + 1) - r_0$ is even.
    Update $\SUFFIXMIN{L_i}$ for the sequences in the new prefix of $\mathcal{L}$ and return a reference to~$\mathcal{L}$.
    
    \item \DELETE$(\SH, e)$ Let $(e',k')=\FINDMIN(\SH)$. If $e\neq e'$, mark item~$e$ to be lazily deleted and leave it in the soft heap. Otherwise, call $\EXTRACTMIN(\SH)$ and return the items becoming corrupted.

    \item \EXTRACTMIN$(\SH)$
    Let $L_i=\SUFFIXMIN{L_1}$, $e=\HEAD{L_i}$, and $k=\KEY{e}$.
    If $\C{e} \neq \emptyset$ we remove the corrupted item $e'=\HEAD{\C{e}}$ from $\C{e}$ and return $(e', k, \emptyset)$, and are done. Otherwise, we corrupt all items in $W(e)$ and prepare to return $(e, k, C)$, where $C=W(e)$. First we remove all items from $C$ that are lazily deleted.
    If $L_i$ has become empty, we remove $L_i$ from~$\mathcal{L}$.
    Otherwise, we update $\SUFFIXMIN{L_{i}}$.
    Finally, we update $\SUFFIXMIN{L_{i-1}},\ldots,\SUFFIXMIN{L_1}$.
    If the new minimum item to be returned by {\FINDMIN} is marked as lazily deleted, we repeatedly remove this, until the new minimum is not lazily deleted or $\mathcal{L}$ is empty, while accumulating all generated corruptions in~$C$.
    Eventually, we return $(e, k, C)$.
\end{itemize}

%%%%%%%%%%%%%%%%%%%%%%%%%%%%%%%%%%%%%%%%
\subsection{Analysis}
%%%%%%%%%%%%%%%%%%%%%%%%%%%%%%%%%%%%%%%%

In the following, we assume for simplicity that all items have distinct keys.
For each item~$e$ in a corruption-set $\C{e'}$ we define an interval~$I(e)$, with
$\KEY{e} \in I(e)$.
If $e$ also has a witness~$e''$, 
i.e.\ $e\in \W{e''}$, we let $I(e) = ]\KEY{e''}, \KEY{e'}]$.
If $e$ has no witness, we let $I(e) = ]-\infty, \KEY{e'}]$.
The interval $I(e)$ captures the partial order maintained by the data structure for item~$e$ in $L_i$.
See Figure~\ref{fig:partial-order}.

That the algorithm maintains a partial order consistent with the total order follows from how the corruption- and witness-sets are updated during {\REDUCE}, and how {\EXTRACTMIN} and {\INSERT} proceed. While we merge sequences $L_i$, we only change the partial order with respect to the items in $L_i$, and these are merged according to the total order. When pruning an item $e_{2i}$ from a sequence, the item $e_{2i}$ keeps $e_{2i-1}$ and $e_{2i+1}$ as predecessor and successor in the partial order. All items in $\W{e_{2i}}$ get their predecessor changed from $e_{2i}$ to $e_{2i-1}$, but since $e_{2i-1}\leq e_{2i}$, the partial order remains valid. Similar all items  $\C{e_{2i}}$ get their successor in the partial order changed from $e_{2i}$ to $e_{2i+1}$, but again the partial order remains valid since  $e_{2i}\leq e_{2i+1}$. When {\EXTRACTMIN} removes the first item~$e$ of an $L_i$, all items in $\W{e}$ loose their lower bound relation to $e$ (and become corrupted), and if {\EXTRACTMIN} returns an item $e' \in \C{e}$, where $e=\min(L_i)$, then the item $e'$ only had a relationship to $e$, and the partial order remains consistent. 
Note that the above also implies that for any item $e$ the interval $I(e)$ can only monotonically increase throughout the lifetime of $e$ in the soft sequence heap. The correctness of the operations, in particular {\FINDMIN} and {\EXTRACTMIN}, follows from the fact that it always returns an item with current key equal to the minimum key of all non-corrupted items in $L_1,\ldots,L_\ell$.

The remaining of this section is devoted to show that the total number of corruptions in a soft sequence heap is bounded by $\epsilon N$ and that the time bounds are as stated in Theorem~\ref{thm:soft-sequence-heap}.

\begin{lemma}
\label{lem:max-rank}
    A sequence with rank~$r$ contains at most $2^r$ items, and
    after $N$ insertions all sequences have rank at most $\floor{\lg N}$.
\end{lemma}
\begin{proof}
    Since a sequence of rank $r$ is the result of merging two sequences of rank $r-1$, a rank zero sequence contains one item, and otherwise items are only removed from a sequence, we by induction have that a sequence of rank~$r$ contains at most $2^r$ items. Furthermore, a rank~$r$ sequence is the result of repeated merging of exactly $2^r$ sequences of rank zero, i.e.\ $2^r$ unique insertions, and we have $2^r\leq N$ implying $r\leq\floor{\lg N}$.
\end{proof}

Let $s_r$ denote an upper bound on the length of a sequence of rank~$r$. By Lemma~\ref{lem:max-rank} we have $s_r\leq 2^r$. The following lemma captures the effect of using {\REDUCE} to prune items.

\begin{lemma}
\label{lem:sr}
    A sequence of rank~$r$ contains at most~$s_r=2^r$ items for $r\leq r_0$, and at most $s_r=(2^{r_0}+1)\cdot 2^{\ceil{(r-r_0)/2}}$ items for $r > r_0$.
\end{lemma}
\begin{proof}
    A sequence of rank zero has size one, and 
    for ranks $1,\ldots,r_0$ a sequence of rank $r$ is the result of merging two sequences of rank $r-1$ without pruning, i.e.\ we have $s_r=2^r$ for $r\leq r_0$.
    For the subsequent ranks, we alternate between just merging two sequences, and merging two sequences followed by {\REDUCE}. 
    The first guarantees $s_{r_0+2p+1}=2\cdot s_{r_0+2p}$, whereas the second guarantees $s_{r_0+2p+2} = s_{r_0+2p+1}+1$
    for $p \geq 0$.
    It follows that for $p\geq 1$, we have
    $s_{r_0+2p} = (\cdots(((2^{r_0}\cdot 2 + 1)\cdot 2 + 1) \cdot 2 + \cdots) \cdot 2 +  1 
    = 2^{r_0} \cdot 2^p + \sum_{i=0}^{p-1} 2^i = (2^{r_0} + 1) \cdot 2^p - 1$.
    The lemma follows since $s_{r_0+2p+1} = s_{r_0+2p+2} - 1$, for $p\geq 0$.
\end{proof}

The following lemma states that the pruning done by {\REDUCE} is quite aggressive, only leaving $o(N)$ items in the sequences. Fortunately, most pruned items will have witnesses and therefore will not be corrupted.

\begin{lemma}
\label{lem:total-size-L}
    For a soft sequence heap the total number of items in $L_1,\ldots,L_\ell$ is 
    $O(\sqrt{N/\epsilon})$.
\end{lemma}
\begin{proof}
    The total number of items in $L_1,\ldots,L_\ell$ is bounded by
    $\sum_{r=0}^{\floor{\lg N}} s_r 
    \leq \sum_{r=0}^{r_0} 2^r + \sum_{r=r_0+1}^{\floor{\lg N}} (2^{r_0} + 1)\cdot 2^{\ceil{(r-r_0)/2}}
%    = O(2^{r_0}\cdot2^{(\lg N-r_0)/2})
    = O(2^{r_0/2}\cdot2^{(\lg N)/2})
    = O(\sqrt{N/\epsilon})$,
    since $r_0=\ceil{\lg\frac{1}{\epsilon}}$.
\end{proof}

\begin{lemma}
\label{lem:total-sequences}
    Over a sequence of heap operations containing $N$ insertions, the total length of all sequences created is bounded by $O(N\lg\frac{1}{\epsilon})$.
\end{lemma}
\begin{proof}
    Over time $N$ insertions can at most create $\floor{N/2^r}$ distinct sequences of rank~$r$ (an item can only once be merged into a sequence of each rank).
    By Lemma~\ref{lem:sr} and  $r_0=\ceil{\lg \frac{1}{\epsilon}}$, the total length of all sequences ever created by $N$ insertions is bounded by
    $$ \sum_{r=0}^{\floor{\lg N}} \floor{\frac{N}{2^r}}\cdot s_r
    =  O\left(\sum_{r=0}^{r_0} \frac{N}{2^r}\cdot 2^r
    + \sum_{r=r_0+1}^{\floor{\lg N}} \frac{N}{2^r}\cdot 2^{(r+r_0)/2}\right)
    = O(N\cdot r_0) = O\left(N\cdot\lg \frac{1}{\epsilon}\right)\;.
    $$
\end{proof}

The following lemma states the amortized running time of the different heap operations. The bound on the number of corrupted items follows by Lemma~\ref{lem:total-corruptions}.

\begin{lemma}
\label{lem:time-bounds}
    Soft sequence heaps support {\INSERT} in amortized $O(\lg \frac{1}{\epsilon})$ time, and the remaining operations in amortized constant time.
\end{lemma}
\begin{proof}
    Over a sequence of heap operations, involving $N$ {\INSERT} operations, a lot of work can only happen once for each inserted item: each item can at most once be pruned, i.e. the total time for pruning and merging corrupted-sets and witness-sets is $O(N)$. Similarly each item can at most be extracted once from a $\C{e}$ set, deleted once from a sequence by an {\EXTRACTMIN} operation, and being reported corrupted at most once. All this work takes total time $O(N)$. The {\MAKEHEAP} and {\FINDMIN} operations clearly take $O(1)$ worst-case time, and do not need to be considered in the following. 
    
    The two sources of non-constant work are when creating new sequences by merging sequences and to update $\SUFFIXMIN{L_i}$ references whenever the minimum item in a sequence changes. The merging of sequences happens during {\INSERT} and {\MELD}. By Lemma~\ref{lem:total-sequences} the total length of all sequences created over time is bounded by $O(N\lg\frac{1}{\epsilon})$. Since creating a sequence $L_i$ by merging (and possibly followed by {\REDUCE}) takes time
    $O(|L_i|)$, the total time for merging sequences is $O(N\lg\frac{1}{\epsilon})$.
    
    The $\SUFFIXMIN{L_i}$ references need to be updated during {\INSERT}, {\MELD} and {\EXTRACTMIN}. 
    During {\INSERT} only $\SUFFIXMIN{L_1}$ needs to be updated, which can be done in constant time.
    If {\EXTRACTMIN} removes and returns the first item in a rank $r$ sequence~$L_i$, then the at most $r+1$ references $\SUFFIXMIN{L_1},\ldots,\SUFFIXMIN{L_i}$ need to be updated, in time $O(r)$.
    There are at most $N$ {\EXTRACTMIN} from sequences of rank $\leq r_0$, each with costs of at most $O(r_0)$, i.e.\ total cost $O(N\cdot r_0)$.
    For each of the $\floor{N/2^r}$ sequences ever created of rank $r>r_0$, at most $s_r$ sequence items can be removed by {\EXTRACTMIN}, each with an update cost of~$O(r)$. The total time for these {\EXTRACTMIN} becomes
    $$O\left(\sum_{r=r_0+1}^{\floor{\lg N}} \floor{\frac{N}{2^r}}\cdot s_r \cdot r\right)
    = O\left(\sum_{r=r_0+1}^{\floor{\lg N}} \frac{N}{2^r}\cdot 2^{(r+r_0)/2}\cdot r\right)    
    = O(N\cdot r_0)\;.$$
    It follows that all {\INSERT} and {\EXTRACTMIN} operations take total time $O(N\cdot r_0)$.
    
    For {\MELD} we need to charge the merging of the $\mathcal{L}_1$ and  $\mathcal{L}_2$ lists, and for updating the $\SUFFIXMIN{L_i}$ references.
    For this we use a separate potential argument. With a soft heap $L_1,\ldots,L_\ell$ we assign a potential $\Phi=\RANK{L_\ell}$, i.e.\ equal to the maximal rank of a sequence. Since {\INSERT} only can increase the maximum rank by one, this only increases the cost of insertions by an additive term.
    For {\MELD} on two soft sequence heaps with sequences with maximum rank $r_1$
    and $r_2$, respectively, the resulting heap will have a sequence with maximal rank/potential at most $\max(r_1, r_2)+1$. I.e.\ potential $\min(r_1, r_2)-1$ is released by {\MELD}. By charging a constant potential to {\MELD}, a total of $O(\min(r_1, r_2))$ released potential will be available for performing the merging of $\mathcal{L}_1$ and  $\mathcal{L}_2$ and for updating the $\SUFFIXMIN{L_i}$ references of old sequences.
    Since at most $N-1$ non-trivial {\MELD} operations can be performed (where both heaps contain at least one item), the total additional cost for handling {\MELD} is $O(N)$.
    
    The total work of the sequence of operations (except {\MAKEHEAP}, {\FINDMIN}, and trivial {\MELD} which take worst-case constant time) is $O(N\lg\frac{1}{\epsilon})$, which can be charged $O(\lg\frac{1}{\epsilon})$ to each insertion, and constant to the remaining operations.
\end{proof}

To bound the number of corruptions, we first bound the size of corruption-sets and witness-sets.
Let $c_r$ and $w_r$ be a bound on the number of items in the corruption-set $\C{e}$ and witness-set $\W{e}$ for $e$ in a sequence~$L_i$ of rank~$r$. 

\begin{lemma}
\label{lem:cr}
    $c_r=w_r=0$ for $r\leq r_0$, and $c_r=w_r=2^{\floor{(r-r_0)/2}}-1$
    for $r>r_0$.
\end{lemma}
\begin{proof}
    Merging $L_i$ sequences does not change $\C{e}$ and $\W{e}$ sets. Only  {\REDUCE}$(L_i)$ add items to $\C{e}$ and $\W{e}$ sets. When pruning $e$ in a sequence $L_i$ with rank $r>r_0$ and $r-r_0$ even, we append $e$ and $\C{e}$ to $\C{e'}$ for the successor $e'$ of $e$ in $L_i$ (and $e$ and $\W{e}$ to $\W{e''}$ for the predecessor $e''$ of $e$ in $L_i$). By only pruning every second item in a sequence, additions to a corruption-set (witness-set) can only come from the predecessor (successor) item in the sequence.
    We have the recurrence
    $$ c_r = \left\{
        \begin{array}{ll}
            0 & \mbox{for } r \leq r_0 \\
            c_{r-1} & \mbox{for } r > r_0 \mbox{ and $r-r_0$ odd} \\
            2 \cdot c_{r-1} + 1 & \mbox{for } r > r_0 \mbox{ and $r-r_0$ even}\;, \\
        \end{array}
    \right. $$
    which solves to $c_r = 2^{\floor{(r - r_0)/2}} - 1$ for $r>r_0$, since
    $c_{r_0 + 2p} = (\cdots((0\cdot 2 + 1)\cdot 2 +1)\cdots)\cdot 2 +1=\sum_{i=0}^{p-1} 2^i = 2^p-1$.
    Similarly we have $w_r = 2^{\floor{(r - r_0)/2}} - 1$.
\end{proof}

For a sequence $L_i$ and a possible key value $x$ we let $D(L_i, x)$ denote the set of items~$e$ in corrupted-sets in $L_i$ where the interval $I(e)$ contains $x$, i.e.\
$$D(L_i, x)=\{e \mid \exists e' \in L_i : e\in\C{e'} \wedge x \in I(e)\}\,.$$
Note that the corrupted items in $L_i$ are exactly $D(L_i, -\infty)$. 
We let $d_r$ denote an upper bound on $|D(L_i, x)|$ for a rank $r$ sequence $L_i$, i.e.\ $d_r$ is an upper bound on the number of corruptions in a sequence of rank~$r$. 

\begin{lemma}
\label{lem:dr}
    $d_r=0$ for $r\leq r_0$, and $d_r=2^{r-r_0-1}$ for $r>r_0$.
\end{lemma}
\begin{proof}
    Since corruption-sets are empty for $r\leq r_0$ we have $d_r=0$ for $r\leq r_0$. 
    When merging two sequences $L_i$ and $L_{i+1}$ of rank $r-1$ we have
    $D(L_i\cup L_{i+1}, x)=D(L_i, x) \cup D(L_{i+1}, x)$, i.e.\ $d_r = 2\cdot d_{r-1}$.
    If we apply {\REDUCE} to the merged sequence, we prune an item $e$ with predecessor $e''$ and successor $e'$. This assigns $I(e)=]\KEY{e''},\KEY{e'}]$, and the at most $c_{r-1}$ items in $\C{e}$ all have their interval extended with $]\KEY{e},\KEY{e'}]$, and the at most $w_{r-1}$ items in $\W{e}$ all have their interval
    extended by $]\KEY{e''},\KEY{e}]$. Since the pruning of every second item ensures that the prunings affect disjoint intervals of the key space, it follows that {\REDUCE} increases $d_r$ additionally by at most $1+\max(c_{r-1},w_{r-1})$. Since $c_{r-1}=w_{r-1}$, we get the following recurrence
    $$ d_r = \left\{
        \begin{array}{ll}
            0 & \mbox{for } r \leq r_0 \\
            2 \cdot d_{r-1} & \mbox{for } r > r_0 \mbox{ and $r-r_0$ odd} \\
            2 \cdot d_{r-1} + c_{r-1} + 1 & \mbox{for } r > r_0 \mbox{ and $r-r_0$ even}\;. \\
        \end{array}
    \right. $$
    Using $c_r=2^{\floor{(r-r_0)/2}}-1$ (Lemma~\ref{lem:cr}), for $r>r_0$,
    the recurrence solves to
    $$d_r=\sum_{i=1}^{\floor{\frac{r-r_0}{2}}} (c_{(r_0+2i)-1} +1 )\cdot 2^{r-(r_0+2i)}
    =\sum_{i=1}^{\floor{\frac{r-r_0}{2}}} (2^{\floor{(r_0+2i-1-r_0)/2}}-1 +1 )\cdot 2^{r-r_0-2i}$$
    $$=\sum_{i=1}^{\floor{\frac{r-r_0}{2}}} 2^{i-1}\cdot 2^{r-r_0-2i}
    =2^{r-r_0-1}\cdot\sum_{i=1}^{\floor{\frac{r-r_0}{2}}} 2^{-i} < 2^{r-r_0-1}\;.$$

    Note that when {\EXTRACTMIN} removes the first item $e$ in $L_i$ this causes all items in $\W{e}$ to loose their witness. But his only happens when $\C{e}=\emptyset$, i.e.\ no interval ends at $\KEY{e}$, and all intervals for items in $\W{e}$ are extended with $]-\infty,\KEY{e}]$. It follows after $e$ is removed $D(L_i, \KEY{e}^-)=D(L_i, \KEY{e}^+)\leq d_r$.
\end{proof}

\begin{lemma}
\label{lem:total-corruptions}
    The total number of corruptions in a soft sequence heap after $N$ insertions is bounded by $\epsilon N$.
\end{lemma}
\begin{proof}
    Recall that the sequences $L_1,L_2,\ldots,L_{\ell}$ have distinct rank and that the maximum rank is bounded by $\floor{\lg N}$. For a sequence $L_i$ of rank $r$, the number of corruptions is $|D(L_i,-\infty)|\leq d_r$, i.e.\ by Lemma~\ref{lem:dr} the total number of corruptions is bounded by
    $$\sum_{r=0}^{\floor{\lg N}} d_r 
    = \sum_{r=r_0+1}^{\floor{\lg N}} 2^{r-r_0-1}
    = \sum_{i=0}^{\floor{\lg N}-r_0-1} 2^i
    = 2^{\floor{\lg N}-r_0}-1
    < N/2^{r_0}
    \leq \epsilon N\;.$$
\end{proof}
Theorem~\ref{thm:soft-sequence-heap} follows from Lemma~\ref{lem:total-corruptions} and Lemma~\ref{lem:time-bounds}.

%%%%%%%%%%%%%%%%%%%%%%%%%%%%%%%%%%%%%%%%
\paragraph*{Remarks}
%%%%%%%%%%%%%%%%%%%%%%%%%%%%%%%%%%%%%%%%

Essential to our construction is that we reduce the length of the merged sequences to avoid spending $\Theta(N\lg N)$ time on merging sequences during $N$ insertions.
The presented solution is based on binary merging and applies {\REDUCE} at every second rank -- inspired by the ``double even fill'' car-pooling used by Kaplan \latin{et al.}~\cite{KaplanTZ13}.
Alternatively, one could increase the merging degree to three (or more) and apply reductions at all ranks~$\geq r_0$.

Comparing our construction to previous constructions, our construction maintains a collection of sorted sequences whereas all previous soft heaps maintain heap ordered binomial trees or binary trees. Similar to our and all previous constructions is the application of car-pooling to achieve the improved performance over (non-soft) heaps and the usage of corruption-sets. Whereas our solution allows a sub-linear number of items not to be stored in corruption-sets (Lemma~\ref{lem:total-size-L}), previous solutions require a larger number of elements not to be stored in corruption-sets, i.e.\ our solution maintains order among a smaller subset of items not in corruption-sets. To be able to report when elements should be considered corrupted, i.e.\ to satisfy the the soft heap interface required by Kaplan~\latin{et al.}~\cite{KaplanKZZ19}, we apply witness-sets.

Note that the witness-sets can be removed completely from the construction if deletions are not required to return the set of items getting corrupted by a deletion --- but witness-sets are still crucial for the analysis to bound the number of corruptions in a soft sequence heap. Interestingly, this implies a structure where only $o(N)$ items are not in corruption-sets, but still guarantees that only $\epsilon N$ keys need to be considered corrupted.

%%%%%%%%%%%%%%%%%%%%%%%%%%%%%%%%%%%%%%%%%%%%%%%%%%%%%%%%%%%%%%%%%%%%%%%%%%%%
\section{Heap ordered ternary tree based soft-heaps}
\label{sec:ternary}
%%%%%%%%%%%%%%%%%%%%%%%%%%%%%%%%%%%%%%%%%%%%%%%%%%%%%%%%%%%%%%%%%%%%%%%%%%%%

Kaplan \latin{et al.}~\cite{KaplanTZ13} describe a soft heap implementation based on a forest of perfectly balanced binary trees, and Kaplan \latin{et al.}~\cite{KaplanKZZ19}  describe how to modify the structure to support the interface described in Section~\ref{sec:soft-heap-interface}. In particular they apply lazy insertions, to circumvent that the original structure might introduce corruptions during insertions. 
Their structure also adopts the notion ``double even fill''. 
In this section we discuss a variant of their structure that avoids both these concepts. The performance remains unchanged, i.e.\ all operations are amortized constant (and independent of $\epsilon$), except {\EXTRACTMIN} and {\DELETE} which require amortized time $O(\lg\frac{1}{\epsilon})$. In the following we skip addressing {\DELETE} (which can be handled by lazy deletions) and {\MELD} (which proceeds very similar as for soft sequence heaps).

The basic structures are perfectly balanced heap-ordered trees. We describe the construction generalized by a degree parameter~$d\geq 3$, although for our result we only need $d=3$.  [This deviates from \cite{KaplanTZ13} that uses $d=2$.] A rank~$r$ tree is a perfectly balanced tree with $d^r$ leaves, where all leaves have depth~$r$ and all internal nodes have $d$ children. Each leaf corresponds to a unique insertion. A tree is kept heap ordered by recursively pulling items up in the tree, leaving subtrees empty (i.e.\ nodes without items), such that the root stores the item with minimum value in the tree.

A simple (non-soft) ``forest heap'' consists of a list $\mathcal{L}$ of trees in non-decreasing rank order, with at most $d-1$ trees of each rank. The insertion of an item~$e$ creates a single node rank zero tree at the front of $\mathcal{L}$, storing~$e$. While the first $d$ trees of $\mathcal{L}$ have equal rank~$r$, we \emph{link} these trees to create a rank~$r+1$ tree: create a new rank~$r+1$ node and make the $d$ rank~$r$ roots the children of this node, and recursively fill the node with an item by moving an item with minimum key from a child one level up, recursively filling the child until no item can be moved up. Filling the new rank $r+1$ root takes time $O(d\cdot(r+1))$. 
During $N$ {\INSERT} operations at most $\floor{\frac{N}{d^r}}$ roots of rank $r$ are created. Since the maximal rank of a tree is $\floor{\lg_d N}$, the total time to link roots during insertions is at most
$O(\sum_{r=1}^{\floor{\lg_d N}} \floor{\frac{N}{d^r}}\cdot d\cdot r)=O(N)$, i.e.\
insertions take amortized constant time.
An {\EXTRACTMIN} operation identifies a root with an item with minimum key, removes this item, and recursively refills a rank~$r$ root in time $O(d\cdot r)$. Since there at most $(d-1)\cdot\floor{\lg_d N}$ roots, an {\EXTRACTMIN} operation takes time $O(d\cdot\lg_d N)$. 

To improve performance, car-pooling is adapted, to avoid moving each item all the way from a leaf to the root. With each item $e$ at a node we store a corruption-set $\C{e}$ of corrupted items $e'$ with $\KEY{e'}\leq \KEY{e}$. Similarly to other soft heap implementations, we maintain suffix-min references for the roots in $\mathcal{L}$.
The implementation of {\INSERT} proceeds as described above for the non-soft case using repeated linking of $d$~trees of equal rank in amortized $O(1)$ time, except that we also need to update in constant time the suffix-min reference for the resulting first tree in $\mathcal{L}$.
[This deviates from the solution in \cite{KaplanTZ13} that uses ``double even fill'' during insertions, which can introduce corruptions, and require the insertions to be buffered and be performed lazily.]

For the implementation of {\EXTRACTMIN} we use a rank threshold
$r_0=\max(2,\ceil{\lg \frac{1}{\epsilon}})$.
We find the root $v$ storing an item $e$ with minimum non-corrupted key in constant time, using the suffix-min reference of the first tree in $\mathcal{L}$. If $\C{e}\neq\emptyset$, we return a corrupted item from $\C{e}$ with current key equal to $\KEY{e}$, without generating any corruptions. Otherwise, $e$ will be returned with its real key. Before doing so, we need to refill $v$ with a new item and update suffix-min references for all roots in $\mathcal{L}$ from right-to-left starting at $v$.
The refilling of an empty node is done with a twist, possibly creating corruptions. 
Whenever a node of rank~$r$ (i.e. the height of the subtree rooted at the node is $r$) is to be refilled, we refill it recursively as in the non-soft case if $r\leq r_0$. If $r>r_0$ recursively move \emph{two} items $e_1$ and $e_2$ with smallest keys from the subtree to the node, $\KEY{e_1}\leq \KEY{e_2}$, and make $e_1$ corrupted (and to be returned as corrupted by {\EXTRACTMIN}) by appending $e_1$ and $\C{e_1}$ to $\C{e_2}$, and leave $e_2$ as the new item at the node. 
[This deviates from the solution in \cite{KaplanTZ13}, that only recursively double fills  for even ranks, which limits the number of courruptions introduced when binary linking is applied.]
The amortized analysis of {\EXTRACTMIN} is given below.

To bound the total number of corruptions, let $c_r$ denote an upper bound on the size of a corruption-set $\C{e}$ for an item $e$ stored at a node of rank~$r$. 
Since corruptions only are introduced at nodes with rank $r>r_0$, we have $c_r=0$ for $r\leq r_0$. The size of $\C{e_2}$ only increases at a node at rank $r$, when $e_2$ and another item $e_1$ are moved up from rank $r-1$ nodes, and $e_1$ and $\C{e_1}$ are appended  to $\C{e_2}$. It follows $c_r =2c_{r-1}+1$ for $r>r_0$, implying $c_r=\sum_{i=0}^{r-r_0+1} 2^i=2^{r-r_0} -1$ for $r>r_0$.
By summing over all possible nodes, the total number of corruptions in a structure is bounded by 
$\sum_{r=0}^{\floor{\lg_d N}} \floor{\frac{N}{d^r}}\cdot c_r
\leq \sum_{r=r_0+1}^{\floor{\lg_d N}} \frac{N}{d^r}\cdot 2^{r-r_0}
< \frac{N}{2^{r_0}}\cdot\left(\frac{2}{d}\right)^{r_0+1}\cdot\frac{d}{d-2}
\leq \frac{N}{2^{r_0}}
\leq \epsilon N$, 
since $r_0=\max(2,\ceil{\lg\frac{1}{\epsilon}})$.

Next we bound the time spend on refilling nodes during $\delta$ {\EXTRACTMIN} operations.
After $N$ insertions there are at most $N/2^{r_0}$ nodes with rank~$>r_0$, 
and $N/2^{r_0}$ items in the corrupted-sets.
Together with the $\delta$ deleted items, a total of at most $\delta+2N/2^{r_0}$
items need to have moved up to nodes with rank $>r_0$.
The recursive pull of an item from rank~$r_0$ to $r_0+1$ takes worst-case $O(d\cdot r_0)$ time,
i.e.\ over all {\EXTRACTMIN} operations we spend $O(d\cdot r_0\cdot(\delta  + N/2^{r_0}))=O(d\cdot N + d \cdot r_0 \cdot \delta)$ time on recursively filling nodes at ranks~$\leq r_0$.
For moving items up to nodes at rank~$>r_0$, we observe that whenever we move two items one level up, one of the items get corrupted --- except for the last possible item being moved into a node before the subtree becomes empty. 
Since an item can at most get corrupted once, and there are at most $N/2^{r_0}$ nodes with rank $>r_0$, at most $O(N)$ times an item is moved one level up. 
We conclude that a sequence with $N$ {\INSERT} and $\delta$ {\EXTRACTMIN} operations in total spend 
$O(d\cdot N+d \cdot r_0 \cdot \delta)$ time on recursively pulling items up during the {\EXTRACTMIN} operations.

Finally,we consider the time to update the suffix-min pointers during {\EXTRACTMIN}. If the item at a root of rank $r$ changes, at most $(d-1)\cdot(r+1)$ suffix-min references need to be updated in time $O(d\cdot r)$. We charge the cost for updating the suffix-min references at the roots of rank $\leq r_0$ directly to the {\EXTRACTMIN} operation, i.e.\ $O(d\cdot r_0)$. 
For updating roots of ranks $r_0 + 1$ to $r$ we just consider a very rough bound on the total number of different items that can become the root of rank~$r$ trees. 
At most $\floor{N/d^r}$ trees are ever created of rank $r$, each such tree contains at most $d^r$ items, of which only a fraction $\Theta(\frac{1}{2^{r-r_0}})$ can reach the root, due to the pruning of every second item reaching nodes of rank $r_0+1,\ldots,r$. In total
$$O\left(\floor{\frac{N}{d^r}} \cdot d^r \cdot \frac{1}{2^{r-r_0}}\right)$$
different items can become the root of rank $r$ trees. Charging updating $O(d\cdot(r-r_0))$ suffix-min references to each of these items, and summing over all ranks gives
$$O\left(\sum_{r=r_0+1}^{\floor{\lg_d N}} N \cdot \frac{1}{2^{r-r_0}} \cdot d\cdot (r-r_0)\right)
=O(N\cdot d)$$
as an upper bound of updating sufffix-min references.

We conclude that a sequence with $N$ {\INSERT} and $\delta$ {\EXTRACTMIN} operations requires total time $O(d\cdot N+d \cdot r_0 \cdot \delta) = O(N+\delta\lg\frac{1}{\epsilon})$  for $d=3$ and $r_0=\max(2,\ceil{\frac{1}{\epsilon}})$, i.e.\ {\INSERT} takes amortized constant time and {\EXTRACTMIN} amortized $O(\lg\frac{1}{\epsilon})$ time.

%%%%%%%%%%%%%%%%%%%%%%%%%%%%%%%%%%%%%%%%%%%%%%%%%%%%%%%%%%%%%%%%%%%%%%%%%%%%
%  Bibliography
%%%%%%%%%%%%%%%%%%%%%%%%%%%%%%%%%%%%%%%%%%%%%%%%%%%%%%%%%%%%%%%%%%%%%%%%%%%%

\bibliographystyle{plainurl}
\bibliography{softheap}

%%%%%%%%%%%%%%%%%%%%%%%%%%%%%%%%%%%%%%%%%%%%%%%%%%%%%%%%%%%%%%%%%%%%%%%%%%%%
%  Appendix
%%%%%%%%%%%%%%%%%%%%%%%%%%%%%%%%%%%%%%%%%%%%%%%%%%%%%%%%%%%%%%%%%%%%%%%%%%%%

\newpage
\begin{appendix}

%%%%%%%%%%%%%%%%%%%%%%%%%%%%%%%%%%%%%%%%%%%%%%%%%%%%%%%%%%%%%%%%%%%%%%%%%%%%
\section{Notation}
%%%%%%%%%%%%%%%%%%%%%%%%%%%%%%%%%%%%%%%%%%%%%%%%%%%%%%%%%%%%%%%%%%%%%%%%%%%% 

\begin{center}
  \begin{tabular}{cl}
    Notation & Description \\
    \hline
    $N$ & total number of insertions \\
    $n$ & current number of items in soft heap \\
    $\epsilon$ & error parameter \\
    $\SH$ & soft heap \\
    $\mathcal{L}$ & set of sequences \\\
    $L_1, L_2, \ldots, L_\ell$ & sorted sequences \\
    $r$ & rank \\
    $r_0$ & rank threshold, $r_0=\ceil{\lg\frac{1}{\epsilon}}$ \\
    $e=(k,v)$ & item = (key, value) pair \\
    $\W{e}$ & Witness-set of non-corrupted item $e$ \\
    $\C{e}$ & Corruption-set of item $e$ \\
    $c_r$ & $|\C{e}|\leq c_r$ for $e$ in rank $r$ sequence \\
    $w_r$ & $|\W{e}|\leq w_r$ for $e$ in rank $r$ sequence \\
    $I(e)$ & corruption interval of item  $e$ \\
    $D(L,x)$ & $D(L,x)=\{e \mid \exists e' \in L: e\in \C{e'} \wedge x\in I(e) \}$ \\
    $d_r$ & $|D(L,x)|\leq d_r$ for rank $r$ sequence $L$ \\
    $s_r$ & $|L|\leq s_r$ for rank $r$ sequence $L$ \\
    $\alpha(m,n)$ & inverse of Ackermann's function \\
    $\delta$ & number of deletions \\
    $d$ & degree parameter \\
    \hline
  \end{tabular}
\end{center}

\end{appendix}

%%%%%%%%%%%%%%%%%%%%%%%%%%%%%%%%%%%%%%%%%%%%%%%%%%%%%%%%%%%%%%%%%%%%%%%%%%%%
\end{document}